\documentclass[letterpaper,11pt]{article}

\usepackage[margin=1in]{geometry}
\usepackage{xspace,exscale,relsize}
\usepackage{fancybox,shadow}
\usepackage{graphicx}
\usepackage{color}
\usepackage{amsthm,amsfonts}
\usepackage{amssymb}
\usepackage{authblk}
\usepackage[title]{appendix}
\usepackage{extarrows}
\usepackage[noadjust]{cite}
\usepackage{amsmath}
\makeatletter
\let\over=\@@over \let\overwithdelims=\@@overwithdelims
\let\atop=\@@atop \let\atopwithdelims=\@@atopwithdelims
\let\above=\@@above \let\abovewithdelims=\@@abovewithdelims
\makeatother
\usepackage{rotating}
\usepackage{ifpdf}
\usepackage{subfigure}
\usepackage{prettyref}
\usepackage{enumitem}
\usepackage{tikz}
\usepackage{pgfplots}
\usepgfplotslibrary{groupplots}
\pgfplotsset{compat=1.18}
\usepackage{float}
\pgfmathdeclarefunction{binentropy}{1}{%
  \pgfmathparse{%
    -(#1)*ln(max(#1,1e-10))/ln(2)
    -(1-(#1))*ln(max(1-(#1),1e-10))/ln(2)%
  }%
}
\usetikzlibrary{arrows}
\tikzstyle{int}=[draw, fill=blue!20, minimum size=2em]
\tikzstyle{dot}=[circle, draw, fill=blue!20, minimum size=2em]
\tikzstyle{init} = [pin edge={to-,thin,black}]
\usepackage[
  CJKbookmarks=true,
  bookmarksnumbered=true,
  bookmarksopen=true,
  colorlinks=true,
  citecolor=red,
  linkcolor=blue,
  anchorcolor=red,
  urlcolor=blue,
  pdftitle={Entropy of Bernoulli Measures Conditioned on Affine Subspaces and a Problem of Ancheta-Massey},
  pdfauthor={}
]{hyperref}
\usepackage[all]{xy}
\usepackage{mathtools}
\usepackage{ifthen}
\newboolean{aos}
\setboolean{aos}{FALSE}

\newcommand{\Law}{\text{Law}}

\newcommand{\bbF}{\ensuremath{\mathbb{F}}}

\ifx\eqref\undefined
  \newcommand{\eqref}[1]{~(\ref{#1})}
\fi
\ifx\mod\undefined
  \def\mod{\mathop{\rm mod}}
\fi

\usepackage{bm}

\def\argmax{\mathop{\rm argmax}}

\def\exp{\mathop{\rm exp}}

\def\Cov{\mathrm{Cov}}

\def\FF{\mathbb{F}}

\newcommand{\reals}{\mathbb{R}}

\newcommand{\Expect}{\mathbb{E}}

\newcommand{\Prob}{\mathbb{P}}
\newcommand{\prob}[1]{\mathbb{P}\left[#1\right]}

\newcommand{\pth}[1]{\left( #1 \right)}
\newcommand{\qth}[1]{\left[ #1 \right]}

\newcommand{\Bern}{\mathrm{Bern}}

\newcommand{\Indc}{\mathbf{1}}

\definecolor{myblue}{rgb}{.8, .8, 1}
\definecolor{mathblue}{rgb}{0.2472, 0.24, 0.6}
\definecolor{mathred}{rgb}{0.6, 0.24, 0.442893}
\definecolor{mathyellow}{rgb}{0.6, 0.547014, 0.24}

\newcommand{\calC}{{\mathcal{C}}}

\newrefformat{eq}{(\ref{#1})}
\newrefformat{thm}{Theorem~\ref{#1}}
\newrefformat{th}{Theorem~\ref{#1}}
\newrefformat{chap}{Chapter~\ref{#1}}
\newrefformat{sec}{Section~\ref{#1}}
\newrefformat{seca}{Section~\ref{#1}}
\newrefformat{algo}{Algorithm~\ref{#1}}
\newrefformat{fig}{Fig.~\ref{#1}}
\newrefformat{tab}{Table~\ref{#1}}
\newrefformat{rmk}{Remark~\ref{#1}}
\newrefformat{clm}{Claim~\ref{#1}}
\newrefformat{def}{Definition~\ref{#1}}
\newrefformat{cor}{Corollary~\ref{#1}}
\newrefformat{lmm}{Lemma~\ref{#1}}
\newrefformat{prop}{Proposition~\ref{#1}}
\newrefformat{pr}{Proposition~\ref{#1}}
\newrefformat{app}{Appendix~\ref{#1}}
\newrefformat{apx}{Appendix~\ref{#1}}
\newrefformat{ex}{Example~\ref{#1}}
\newrefformat{exer}{Exercise~\ref{#1}}
\newrefformat{soln}{Solution~\ref{#1}}

\def\unifto{\mathop{{\mskip 3mu plus 2mu minus 1mu%
  \setbox0=\hbox{$\mathchar"3221$}%
  \raise.6ex\copy0\kern-\wd0%
  \lower0.5ex\hbox{$\mathchar"3221$}}\mskip 3mu plus 2mu minus 1mu}}

\ifx\lesssim\undefined
\def\simleq{{{\mskip 3mu plus 2mu minus 1mu%
  \setbox0=\hbox{$\mathchar"013C$}%
  \raise.2ex\copy0\kern-\wd0%
  \lower0.9ex\hbox{$\mathchar"0218$}}\mskip 3mu plus 2mu minus 1mu}}
\else
\def\simleq{\lesssim}
\fi

\ifx\gtrsim\undefined
\def\simgeq{{{\mskip 3mu plus 2mu minus 1mu%
  \setbox0=\hbox{$\mathchar"013E$}%
  \raise.2ex\copy0\kern-\wd0%
  \lower0.9ex\hbox{$\mathchar"0218$}}\mskip 3mu plus 2mu minus 1mu}}
\else
\def\simgeq{\gtrsim}
\fi

\newtheorem{theorem}{Theorem}
\newtheorem{lemma}{Lemma}
\newtheorem{corollary}{Corollary}

\theoremstyle{definition}

\renewcommand{\hat}{\widehat}

\title{
Entropy of Bernoulli Measures Conditioned on Affine Subspaces and a Problem of Ancheta--Massey
}

\author{Yihong Wu\thanks{Department of Statistics and Data Science, Yale University, New Haven, Connecticut, USA, \texttt{yihong.wu@yale.edu}.}}
	
\date{\today}

\begin{document}
\ifpdf
\DeclareGraphicsExtensions{.pgf}
\graphicspath{{figures/}{plots/}}
\fi
\maketitle

\begin{abstract}
A textbook result in information theory is that linear encoders achieve the entropy for lossless compression of Bernoulli source with parameter $p$. For lossy compression, however, linearity is known to incur strict suboptimality compared to the rate-distortion function. Massey asked whether the optimal rate for linear encoding is achieved simply by compressing a fraction of the bits linearly and losslessly and estimating the rest by zero \cite{Massey1978}. For $p=\frac12$, Ancheta answered this question affirmatively \cite{Ancheta1978}. This note extends Ancheta's result to all $p<\frac12$. 

The key argument is to bound the entropy of the posterior distribution conditioned on an affine subspace in terms of its marginals. The proof was discovered by GPT-5.6 Sol in an interactive process guided by the author. The purpose of the present note is to communicate a simplified version of this proof and to make connections with the existing literature on coding theory and spin glass theory.
\end{abstract}

\section{Introduction}\label{sec:introduction}

A textbook result in information theory is that entropy, the fundamental limit of data compression \cite{shannon}, can be achieved using a linear encoder (see, e.g., \cite[Problem 1.7]{ckbook2}). 
Specifically, let $X \in \bbF_2^n$ have iid coordinates drawn from $\Bern(p)$. There exists a matrix $H \in \bbF_2^{k\times n}$ and a decoder $g: \bbF_2^k\to\bbF_2^n$, such that $\prob{g(HX)\neq X} \to 0$ as $n\to\infty$ and $k = n h(p)+o(n)$. Here 
\[
h(p) \triangleq  p \log_2 \frac{1}{p} + (1-p) \log_2 \frac{1}{1-p}
\]
is the Shannon entropy of $\Bern(p)$ in bits.
Without loss of generality, we assume $p\leq \frac{1}{2}$ throughout the paper.

For lossy data compression, Shannon's rate-distortion theory determines the optimal tradeoff between the compression rate and the reconstruction error. 
For the Bernoulli source $X$, suppose that on average a fraction $D$ of the bits are allowed to be erroneously decoded. (The distortion parameter $D\leq p$ is known as the bit error rate.)
There exist $f:\bbF_2^n\to\bbF_2^k$ and $g:\bbF_2^k\to\bbF_2^n$, such that 
the expected Hamming distance between the source and the reconstruction satisfies
$\Expect[d_H(X,g(f(X))] \leq nD + o(n)$,
and
$k = n R(D) + o(n)$,
where
\begin{equation}
R(D) \triangleq h(p) - h(D)
\label{eq:RD}
\end{equation}
is the \emph{rate-distortion function}, which is the minimal compression rate for any compression scheme achieving bit error rate $D$.

Unlike lossless data compression, it is known that linearity\footnote{Here linear encoding is not to be confused with the scheme of using a linear subspace $C$ as codebook and (nonlinearly) encoding $X$ by its closest point $\hat X$ in $C$. While this scheme is known to achieve the rate-distortion function $1-h(D)$ for $p=\frac{1}{2}$ \cite{goblick1963coding} (see also \cite[Theorem 6.2.2]{berger} and \cite{DelsartePiret1986}), it is strictly suboptimal for $p<\frac{1}{2}$. To see this, suppose $C=\ker(A)$ for some $A\in\bbF_2^{(n-k)\times n}$ with full rank, which has a submatrix $A_I$ formed by $n-k$ columns indexed by $I$. Since $AX=A(X-\hat X)$, $nh(D) \geq H(X-\hat X) \geq H(AX) \geq H(A_I X_I)=H(X_I)=(n-k) h(p)$, the last inequality applying the concavity of entropy. Thus the compression rate satisfies $k/n\geq 1-\frac{h(D)}{h(p)}=\frac{R(D)}{h(p)}$. For modifications of this scheme that achieve the rate-distortion function, see \cite{ChenHeJagmohan2007,GuptaVerdu2009}.
%, which is $\frac{1}{h(p)}$ times the rate-distortion function.
} incurs a strict suboptimality for lossy compression. 
A result due to Ancheta \cite{Ancheta1978,Massey1978} is that if the compressor is restricted to a linear map, the bit error rate must be at least
\begin{equation}\label{eq:ancheta-converse}
  D
  \geq
	\begin{cases}
		(1-R)h^{-1}\left(\frac{h(p)-R}{1-R}\right),
  & p<\frac{1}{2} \\
	 \frac{1-R}{2} & p=\frac{1}{2}.
	\end{cases}  
\end{equation}
where $R=k/n \leq h(p)$ is the compression rate and $h^{-1}$ takes values in $[0,1/2]$.
This shows that linear encoding fails achieve the rate-distortion function. (See \prettyref{fig:rate-distortion-curves}.)

For the special case $p=\frac{1}{2}$, Ancheta's lower bound $D=\frac{1-R}{2}$ may be achieved by the (trivial) linear encoder that keeps the first $R n$ bits and randomly guesses the remaining $(1-R)n$ bits. 
For general $p$, this may be extended to a \emph{time-sharing} scheme that linearly compresses a $\frac{R}{h(p)}$ fraction of the bits losslessly and decodes the remaining bits by 0, achieving an overall distortion
%mixes lossless compression and totally lossy compression.
\begin{equation}
D=p\pth{1-\frac{R}{h(p)}}.
\label{eq:time-share}
\end{equation}
In \cite{Massey1978,Massey2009}, Massey posed the question whether this time-sharing bound is the best one can do with linear encoders.
This question has since remained open and led to related conjectures in the literature \cite{MazumdarPal2017,MazumdarPal2021} (see also \cite[Chap.~9]{Huang2015} for a discussion of linear encoding in lossy compression.)
In this paper we resolve this question in the positive by showing that \prettyref{eq:time-share} is tight for linear encoding for all $p<\frac{1}{2}$.

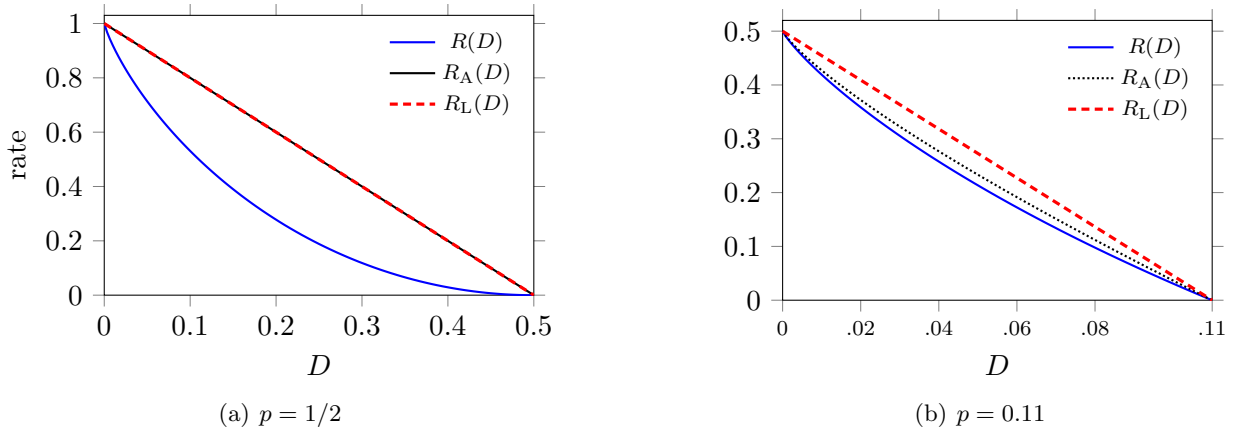
\begin{figure}[ht]
\centering
\pgfmathsetmacro{\hbiased}{binentropy(0.11)}
\subfigure[$p=1/2$]{%
\begin{tikzpicture}
\begin{axis}[
  width=0.44\textwidth,
  height=0.32\textwidth,
  xmin=0,
  ymin=0,
  xlabel={$D$},
  ylabel={rate},
  tick align=outside,
  xmax=0.5,
  ymax=1.03,
  xtick={0,0.1,0.2,0.3,0.4,0.5},
  ytick={0,0.2,0.4,0.6,0.8,1},
  legend style={
    draw=none,
    fill=none,
    font=\scriptsize,
    at={(0.98,0.98)},
    anchor=north east
  }
]
\addplot[
  blue,
  thick,
  domain=0:0.5,
  samples=151
] {max(0,1-binentropy(x))};
\addlegendentry{$R(D)$}

\addplot[
  black,
  thick,
  domain=0:0.5
] {1-2*x};
\addlegendentry{$R_{\mathrm A}(D)$}

\addplot[
  red,
  very thick,
  densely dashed,
  domain=0:0.5
] {1-2*x};
\addlegendentry{$R_{\mathrm L}(D)$}
\end{axis}
\end{tikzpicture}%
\label{fig:rate-distortion-unbiased}%
}
\hfill
\subfigure[$p=0.11$]{%
\begin{tikzpicture}
\begin{axis}[
  width=0.44\textwidth,
  height=0.32\textwidth,
  xmin=0,
  ymin=0,
  xlabel={$D$},
  tick align=outside,
  legend style={
    draw=none,
    fill=none,
    font=\scriptsize,
    at={(0.98,0.98)},
    anchor=north east
  },
  xmax=0.11,
  ymax=0.52,
  xtick={0,0.02,0.04,0.06,0.08,0.11},
  xticklabels={$0$,$.02$,$.04$,$.06$,$.08$,$.11$},
  xticklabel style={font=\scriptsize},
  ytick={0,0.1,0.2,0.3,0.4,0.5}
]

\addplot[
  blue,
  thick,
  domain=0:0.11,
  samples=151
] {max(0,\hbiased-binentropy(x))};
\addlegendentry{$R(D)$}

\addplot[
  black,
  thick,
  densely dotted,
  domain=0:0.11,
  samples=151,
  variable=\t
]
(
  {(1-\hbiased)*\t/(1-binentropy(\t))},
  {(\hbiased-binentropy(\t))/(1-binentropy(\t))}
);
\addlegendentry{$R_{\mathrm A}(D)$}

\addplot[
  red,
  very thick,
  densely dashed,
  domain=0:0.11
] {\hbiased*(1-x/0.11)};
\addlegendentry{$R_{\mathrm L}(D)$}
\end{axis}
\end{tikzpicture}
\label{fig:rate-distortion-biased}%
}
\caption{Rate-distortion tradeoffs for a $\Bern(p)$ source with Hamming loss: $R(D)$ is Shannon's rate-distortion function \prettyref{eq:RD}, $R_{\mathrm A}(D)$ is Ancheta's converse \prettyref{eq:ancheta-converse}, and $R_{\mathrm L}(D)$ corresponds to the time-sharing bound \prettyref{eq:time-share}.}
\label{fig:rate-distortion-curves}
\end{figure}

\section{Problem formulation and main results}\label{sec:problem}

Let
\[
  X=(X_1,\ldots,X_n)^\mathsf T,
  \qquad X_i\stackrel{\mathrm{iid}}{\sim}\Bern(p),
  \qquad 0<p\leq\frac12.
\]
A linear compressor is a matrix $H\in\FF_2^{k\times n}$ that 
maps $X$ to $HX$ from $n$ bits to $k$ bits. (Here and below, vector-matrix multiplication is in $\bbF_2$.)
Because redundant linear projections provide no extra information, we may assume without loss of generality that $H$ has full row rank.

Given an encoder, the optimal decoder that minimizes the Hamming loss is the so-called \emph{bitwise maximum a posteriori} (MAP) decoder that reports the mode of each posterior marginal. 
That is, upon observing $HX=s$, the $i$th decoded bit is
\[
\hat X_i(s) \in \argmax_{b \in \{0,1\}} \Prob(X_i=b|HX=s)
\]
resulting in the minimal Hamming loss
\begin{equation}\label{eq:affine-error}
\Expect[d_H(\hat X,X)|HX=s]	= \sum_{i=1}^n\min\{\Prob(X_i=0|HX=s), \Prob(X_i=1|HX=s)\}.
\end{equation}
The optimal bit error rate for a given linear encoder $H$ is thus
\begin{equation}\label{eq:distortion}
  D(H)\triangleq \frac1n\Expect\qth{
	\sum_{i=1}^n\min\{\Prob(X_i=0|HX), \Prob(X_i=1|HX)\}}.
\end{equation}

The following theorem gives a nonasymptotic converse for linear encoding, proving the tightness of the time-sharing bound \prettyref{eq:time-share}.
\begin{theorem}
\label{thm:finite-converse}
For any $H\in\FF_2^{k\times n}$ with full row-rank,
 %and the optimal reconstruction from $S=HX$,
\begin{equation}\label{eq:finite-converse}
  \frac{k}{n}
  \geq
  h(p)\left(1-\frac{D(H)}{p}\right).
\end{equation}
\end{theorem}

\subsection{Conditional entropy on a coset}\label{sec:affine-entropy}

It will be more convenient to phrase the discussion in terms of affine subspaces or cosets.
Let $C$ be a linear subspace of $\FF_2^n$ and let $A=a+C$ be an affine subspace for $a\in \bbF_2^n$.
Then $A$ is a coset of the subgroup $C$ of $(\mathbb F_2^n,+)$.

Throughout, $P=\Law(X)$ denotes the product Bernoulli measure with parameter $p$ on
the Hamming cube, with the ambient dimension $n$ suppressed to ease the notation.
The probability mass function of $P$  is
\begin{equation}\label{eq:bernoulli-measure}
  P(x)
  =p^{|x|}(1-p)^{n-|x|},
  \qquad x\in\FF_2^n,
\end{equation}
where $|x|$ denotes the Hamming weight.

Given $X\in A$, denote the conditional distribution by
$P_A\triangleq \Law(X\mid X\in A)$, that is,
\begin{equation}\label{eq:posterior}
  P_A(x)=P(x\mid A) = \frac{P(x)}{P(A)},
  \qquad x\in A.
\end{equation}
Given $X \in A$, define the minimal Hamming loss by
\begin{equation}\label{eq:eA}
e(A) \triangleq \min_{\hat x \in \bbF_2^n} \Expect[d_H(X,\hat x) \mid X \in A] = 
\sum_{i=1}^n\min\{P_A(X_i=0),P_A(X_i=1)\},
%\sum_{i=1}^n\min\{q_i(A),1-q_i(A)\},
\end{equation}
attained by the bitwise MAP rule
\[
\hat x_i = \argmax_{b \in \{0,1\}} P_A(X_i=b).
\]
As such, $e(A)$ may be interpreted as the \emph{radius} of $A$ weighted by the Bernoulli measure.

The following key lemma bounds the entropy of the conditional distribution of a Bernoulli vector on an affine subspace in terms of its marginals.
\begin{lemma}
\label{lmm:affine-entropy}
For every affine subspace $A\subseteq\FF_2^n$,
\begin{equation}\label{eq:affine-entropy}
  H(P_A)
  \leq
  \frac{h(p)}{p}\,e(A).
\end{equation}
\end{lemma}

Assuming \prettyref{lmm:affine-entropy}, the proof of the main result is immediate.

\begin{proof}[Proof of \prettyref{thm:finite-converse}]
For $s\in\FF_2^k$, define the affine space
\[
  A_s\triangleq\{x\in\FF_2^n:Hx=s\}
\]
which is non-empty because $H$ has full rank.

Conditioned on $S\equiv HX=s$, $X$ is distributed as $P_{A_s}$. Applying \prettyref{lmm:affine-entropy} yields the pointwise bound
\[
  H(X\mid S=s)
  \leq\frac{h(p)}{p} e(A_s).
\]
Averaging over $S$ and applying the definition \eqref{eq:distortion}, we get
\[
  H(X\mid S)
  \leq\frac{h(p)}{p}\Expect[e(A_S)]
  =\frac{h(p)}{p}\,nD(H).
\]
Since $S$ is a deterministic function of $X$ and takes values in $\FF_2^k$,
\[
  k\geq H(S)=H(X)-H(X\mid S)
  \geq nh(p)-\frac{h(p)}{p}nD(H).
\]
Dividing by $n$ proves \eqref{eq:finite-converse}.
\end{proof}

\section{Proof of \prettyref{lmm:affine-entropy}}
\label{sec:affine-comparison}

Let $X\sim P$.
Fix $A=a+C$ and choose an optimal (bitwise MAP) reconstruction point $b\in\FF_2^n$ for $P_A$, so that
\begin{equation}\label{eq:MAP-center}
  e(A)=\Expect_{P_A}[d_H(X,b)] = \Expect_{P_A}[|X+b|].
\end{equation}
Define the product Bernoulli measure shifted by $b$ as
 $P^b \triangleq \Law(X+b)$, namely,
\begin{equation}\label{eq:centered-product}
  P^b(x)
  = P(x+b) = p^{d_H(x,b)}(1-p)^{n-d_H(x,b)}.
\end{equation}
The key estimate is the following chain of inequalities:
\begin{equation}\label{eq:affine-mass-chain}
  P^b(A)=P(a+b+C)\leq P(C)
  \leq(1-p)^{n-e(C)/p}\leq(1-p)^{n-e(A)/p}.
\end{equation}
With this, we may complete the proof of \prettyref{lmm:affine-entropy} as follows.
Applying a change of measure from $P_A$ to $P^b$ and using
\eqref{eq:MAP-center}--\prettyref{eq:centered-product} to evaluate the cross entropy, we get
\[
  H(P_A)
  =\underbrace{n\log_2\frac1{1-p}
   +e(A)\log_2\frac{1-p}{p}}_{-\Expect_{X\sim P_A}[\log_2 P^b(X)]}
   -D_{\rm KL}(P_A\Vert P^b).
\]
where $D_{\rm KL}$ denotes the KL divergence in bits.
Since $P_A$ is supported on $A$, applying the data processing inequality to the indicator
$x\mapsto 1\{x\in A\}$ yields
\[
  D_{\rm KL}(P_A\Vert P^b)
  \geq\log_2\frac1{P^b(A)}.
\]
Combining these two displays with the final inequality in
\eqref{eq:affine-mass-chain}, we obtain
\[
  H(P_A)
  \leq e(A)\left[
     \log_2\frac{1-p}{p}
     +\frac1p\log_2\frac1{1-p}
   \right]
  =\frac{h(p)}p e(A),
\]
which is \eqref{eq:affine-entropy}. Thus it remains to prove
\eqref{eq:affine-mass-chain}.

\subsection{Reduction to zero coset}

The left and right inequalities in  \eqref{eq:affine-mass-chain} allow us to reduce an arbitrary coset to the zero coset. These are established by the following two lemmas.

\begin{lemma}\label{lmm:zero-mode}
Let $N\geq1$, and let $W=(W_1,\ldots,W_N)^\mathsf T$ have independent coordinates with
\[
  W_j\sim\Bern(r_j),
  \qquad 0< r_j\leq\frac12,
\]
and let $H\in\FF_2^{k\times N}$ be arbitrary. Define
$g(s)\triangleq\Prob(HW=s)$.
%\[
  %g(s)\triangleq\Prob(HW=s),
  %\qquad s\in\FF_2^k.
%\]
Then 
%for every $s\in\FF_2^k$
\begin{equation}\label{eq:zero-mode}
  %g(s)\leq g(0).
	g(0) = \max_{s\in\FF_2^k} g(s).
\end{equation}
Furthermore, for all $s,t\in\FF_2^k$,
% Reserve equation number 19 for the Fourier identity in the proof below.
\begin{equation}\label{eq:four-point}
  g(s)g(t) \leq g(0)g(s+t).
\end{equation}
\end{lemma}

\begin{proof}
Write $h_j$ for the $j$th column of $H$.
Recall the Fourier inversion on $\FF_2^k$ (with respect to the uniform distribution):
for any $f:\FF_2^k\to\reals$,
\[
f(s) = \sum_{u \in \FF_2^k} \hat f(u) (-1)^{u\cdot s}, \quad \hat f(u) \triangleq 2^{-k}
\sum_{s \in \FF_2^k} f(s) (-1)^{u\cdot s}.
\]
Applying this to the function $g$, whose Fourier coefficient is
\[
\hat g(u) =
2^{-k}
\sum_{s \in \FF_2^k} \Prob(HW=s) (-1)^{u\cdot s} =
2^{-k} \Expect\left[(-1)^{u\cdot HW}\right]
= 2^{-k} \prod_{j:u\cdot h_j=1}(1-2r_j),
\]
we get
\begin{align}
  g(s)
  &=2^{-k}\sum_{u\in\FF_2^k}
    (-1)^{u\cdot s}
    \prod_{j:u\cdot h_j=1}(1-2r_j).
  \label{eq:inhomogeneous-syndrome-fourier}
\end{align}
Since every product in the last sum is nonnegative, replacing $s$ by $0$ yields an upper bound, proving \eqref{eq:zero-mode}.

To prove \eqref{eq:four-point}, let $W'$ be an independent copy of $W$, and set $Z\triangleq W+W'$. Because the pairs $(W_j,W_j')$ are independent across $j$, the coordinates of $W$ are independent conditioning on any $Z=z$, with
\begin{equation}\label{eq:conditional-copy-parameters}
  \Prob(W_j=1\mid Z=z)
  =
  \begin{cases}
    \displaystyle\frac{r_j^2}{r_j^2+(1-r_j)^2},&z_j=0\\[1.2ex]
    \displaystyle\frac12,&z_j=1
  \end{cases}
	\quad\leq \frac{1}{2}.
\end{equation}
Moreover, independence of $W$ and $W'$ gives
\begin{align*}
  g(0)g(s+t)
  &=\Prob(HW=0,HW'=s+t)
   =\Prob(HW=0,HZ=s+t),\\
  g(s)g(t)
  &=\Prob(HW=s,HW'=t)
   =\Prob(HW=s,HZ=s+t).
\end{align*}
Consequently,
\[
  g(0)g(s+t)-g(s)g(t)
  =\sum_{z:Hz=s+t}\Prob(Z=z)
    \bigl[
      \Prob(HW=0\mid Z=z)
      -\Prob(HW=s\mid Z=z)
    \bigr]\geq0,
\]
where the inequality follows from \eqref{eq:zero-mode} applied to the conditional product law in \eqref{eq:conditional-copy-parameters} for each $z$.
\end{proof}

%Let $q_i(A)$ denote the marginal of the posterior $P_A = \Law(X|X\in A)$.
Let $q_i(A) = P(X_i=1|X=A)$ denote the $i$th marginal of the posterior.
\begin{lemma}\label{lmm:zero-min-error}
Let $C$ be a linear subspace of $\FF_2^n$.
Then, for every coordinate $i$ and every
$a\in\FF_2^n$,
\begin{equation}\label{eq:q-zero-le-p}
  q_i(C)\leq p,
\end{equation}
\begin{equation}\label{eq:coordinate-extremality}
  q_i(C)\leq q_i(a+C)\leq1-q_i(C).
\end{equation}
Consequently,
\begin{equation}\label{eq:zero-min-error}
  e(C)\leq e(a+C)
\end{equation}
and 
\begin{equation}
e(C) = \Expect_{P_C}[|X|].
\label{eq:bitmap-zero}
\end{equation}
\end{lemma}

\begin{proof}
The assertion is immediate when $n=1$, so assume $n\geq2$.  Choose a
full-row-rank matrix $H$ with $C=\ker H$, 
and define the coset
\[
  A_s\triangleq\{x\in\FF_2^n:Hx=s\}.
\]
Because $H$ has full row rank, every $A_s$ is nonempty.
With this notation, we have $C=A_0$ and $a+C = A_{Ha}$.
Define the syndrome PMF
\[
  f(s)\triangleq\Prob(HX=s)=P(A_s)>0.
\]
For fixed $i$, let $h_i$ denote the $i$th column of $H$ and $H_{-i}$ the submatrix with $h_i$ removed. Similarly, let $X_{-i} = (X_j: j \neq i)$.
Denote the PMF of $H_{-i}X_{-i}$ by
\[
  %Y_i\triangleq\sum_{j\neq i}h_jX_j,
  %\qquad
  g_i(r)\triangleq\Prob(H_{-i}X_{-i}=r),
\]
Conditioning on $X_i$ yields
\begin{equation}\label{eq:syndrome-pmf-deleted}
	f(s)
	%\Prob(HX=s)
	= \Prob(H_{-i}X_{-i}=s,X_i=0) + \Prob(H_{-i}X_{-i}=s+h_i,X_i=1) = (1-p)g_i(s)+p g_i(s+h_i).
\end{equation}

Recall that $q_i(A_s)=P_{A_s}(X_i=1)$. Bayes' rule gives\footnote{If $g_i(h_i)=0$, then $q_i(C)=p g_i(h_i)/f(0)=0$, so
\eqref{eq:q-zero-le-p} and \eqref{eq:coordinate-extremality} are immediate.
We may therefore assume that $g_i(h_i)>0$. For every $s$, applying
\eqref{eq:four-point} to the pairs $(s,h_i)$ and $(s+h_i,h_i)$ shows that
$g_i(s)=0$ if and only if $g_i(s+h_i)=0$. Since $f(s)>0$, both quantities
are positive, and all the ratios in \prettyref{eq:coordinate-posterior-deleted}--\prettyref{eq:coordinate-posterior-deleted0} are well-defined.}
\begin{equation}\label{eq:coordinate-posterior-deleted}
  q_i(A_s)
	=\frac{p }{p+(1-p) \frac{g_i(s)}{g_i(s+h_i)}}
  %\qquad
  %1-q_i(A_s)=\frac{(1-p)g_i(s)}{(1-p)g_i(s)+p g_i(s+h_i)}.
\end{equation}
and in particular for $s=0$,
\begin{equation}\label{eq:coordinate-posterior-deleted0}
  %q_i(C)=\frac{p g_i(h_i)}{(1-p)g_i(0)+p g_i(h_i)}
	q_i(C)=\frac{p }{p+(1-p) \frac{g_i(0)}{g_i(h_i)}}
  %\qquad
  %1-q_i(C)=\frac{(1-p)g_i(0)}{(1-p)g_i(0)+p g_i(h_i)}.
\end{equation}
Applying \prettyref{lmm:zero-mode} to
the matrix $H_{-i}$, we get $g_i(h_i)\leq g_i(0)$ and hence $q_i(C)\leq p$, proving
\eqref{eq:q-zero-le-p}.

To show the two-sided bound \prettyref{eq:coordinate-extremality}, applying \eqref{eq:four-point} to $H_{-i}$ gives
\begin{align*}
	\frac{g_i(h_i)}{g_i(0)}  \leq \frac{g_i(s)}{g_i(s+h_i)} \leq \frac{g_i(0)}{g_i(h_i)}
\end{align*}
The second inequality applied to \eqref{eq:coordinate-posterior-deleted}--\eqref{eq:coordinate-posterior-deleted0} yields $q_i(A_s)\geq q_i(C)$.
For the other direction, the first
inequality is equivalent to
\[
\frac{(1-p)^2}{p^2} \leq \frac{1-q_i(A_s)}{q_i(A_s)} \frac{1-q_i(C)}{q_i(C)}.
\]
Since $p\leq1/2$, replacing the LHS by 1 yields $q_i(A_s)\leq1-q_i(C)$.
%Therefore
%\[
  %q_i(C)\leq q_i(A_s)\leq1-q_i(C).
%\]
%Since $a+C=A_{Ha}$,
This proves \eqref{eq:coordinate-extremality}.

Finally, \eqref{eq:coordinate-extremality} is equivalent to
$\min\{q_i(a+C),1-q_i(a+C)\}\geq q_i(C)$. 
Summing over $i$ proves \prettyref{eq:zero-min-error}.
By \prettyref{eq:q-zero-le-p}, $q_i(C)\leq p \leq 1/2$ so the bit-MAP decoder outputs 0
given $X \in C$ and hence \prettyref{eq:bitmap-zero}.
\end{proof}

\subsection{Bounding the mass of zero coset}

The following lemma bounds the Bernoulli measure of a linear subspace by its radius.
Note that the upper bound is tight for the extreme case of $C=\{0\}$ and $C=\bbF_2^n$.
\begin{lemma}\label{lmm:mass-energy}
For every linear subspace $C$ of $\FF_2^n$,
\begin{equation}\label{eq:mass-energy}
  P(C)
  \leq
  (1-p)^{n-e(C)/p}.
\end{equation}
\end{lemma}

\begin{proof}
Let $\beta\triangleq D_{\rm KL}\bigl(\delta_0\Vert
  \Bern(p)\bigr)
  =\log_2\frac1{1-p}$. 
Define
\[
  J(C)\triangleq D_{\rm KL}\bigl(\delta_{0}\Vert P_C\bigr)
  =-\log_2 P_C(0)
	= -\log_2 \frac{(1-p)^n}{P(C)} = n \beta+\log_2 P(C).
\]
Thus the desired \eqref{eq:mass-energy} is equivalent to
\begin{equation}\label{eq:zero-atom-divergence-bound}
  J(C)\leq\frac{\beta}{p}e(C).
\end{equation}

We prove \eqref{eq:zero-atom-divergence-bound} by induction on $n$.  When
$n=1$, the two possible linear subspaces are $\{0\}$ and $\FF_2$, and both
give equality.  Assume $n\geq2$, and define
\[
  C_t\triangleq\{z\in\FF_2^{n-1}:(z,t)\in C\},
  \qquad t\in\{0,1\}.
\]
Let
\[
  q\triangleq P_C(X_n=1).
\]
By \prettyref{lmm:zero-min-error}, $q\leq p$.

Suppose first that $C_1$ is nonempty.  Then $C_0$ is linear and $C_1$ is an
affine coset of $C_0$.  Let $Z$ denote the first $n-1$ coordinates of $X$ and
define
\[
  m_t\triangleq\Expect_{P_{C_t}}[|Z|],
  \qquad t\in\{0,1\}.
\]
Due to the marginal independence of $Z$ and $X_n$, 
under $P_C$, the conditional law of $Z$ given $X_n=t$ is $P_{C_t}$.
Consequently, 
\begin{align*}
  e(C)
  &=\Expect_{P_C}|X|=\Expect_{P_C}|Z|+P_C(X_n=1)\\
  &=P_C(X_n=0)\Expect_{P_{C_0}}|Z|
    +P_C(X_n=1)\Expect_{P_{C_1}}|Z|+P_C(X_n=1)
    =(1-q)m_0+qm_1+q.
\end{align*}
Applying \prettyref{eq:zero-min-error} and \prettyref{eq:bitmap-zero} in 
\prettyref{lmm:zero-min-error} to the linear space $C_0$ and its
affine coset $C_1$ yields
%and then summing over the coordinates, gives
\[
  %m_0=e(C_0),
  %\qquad
  m_1\geq e(C_1) \geq e(C_0) =  m_0.
\]
Thus 
\begin{equation}\label{eq:expected-weight-error-recursion}
  e(C)\geq e(C_0)+q.
\end{equation}
If $C_1$ is empty, then $q=0$, $C=C_0\times\{0\}$, and
$e(C)=e(C_0)$, so \eqref{eq:expected-weight-error-recursion} also holds.

In order to apply induction, we apply the KL chain rule to obtain
\begin{equation}\label{eq:zero-atom-KL-recursion}
  J(C)
  =J(C_0)+D_{\rm KL}\bigl(\delta_0\Vert
    \Bern(q)\bigr).
\end{equation}
Because $q\leq p$, by convexity of KL divergence, 
\begin{equation}
%\log_2 \frac{1}{1-q} = 
D_{\rm KL}\bigl(\delta_0\Vert
    \Bern(q)\bigr) \leq 
		\frac{q}{p} D_{\rm KL}\bigl(\delta_0\Vert
    \Bern(p)\bigr)
  =\frac{q\beta}{p}.
\label{eq:smallKL}
\end{equation}
Thus 
\[
  J(C)\leq J(C_0)+\frac{q\beta}{p}
  \leq\frac{\beta}{p}\bigl(e(C_0)+q\bigr)
  \leq\frac{\beta}{p}e(C),
\]
where the middle inequality applies the induction hypothesis for $n-1$ to $C_0$,
and the last inequality is \eqref{eq:expected-weight-error-recursion}.
This completes the proof.
\end{proof}

Finally, note that the left, middle, and right inequality in \eqref{eq:affine-mass-chain} follows from
 \prettyref{lmm:zero-mode}, \prettyref{lmm:mass-energy}, and
\prettyref{lmm:zero-min-error}, respectively. This completes the proof of
\prettyref{lmm:affine-entropy}.

\section{Discussion}
\label{sec:discussions}

\subsection{Connections with the literature}
\label{sec:literature-connections}

While \prettyref{lmm:affine-entropy} may appear new, several of its supporting results
%, namely, \prettyref{lmm:zero-mode}--\prettyref{lmm:mass-energy}, 
have appeared in or can be deduced from the existing literature on coding theory and spin glass theory.
\begin{itemize}
	\item 
	For \prettyref{eq:zero-mode}, 
	in the homogeneous case (all $r_j$ are the same), 
	Sullivan \cite{Sullivan1967} showed that zero is the mode of the syndrome $HW$.
	(In fact, a stronger lower bound on $g(0)/g(s)$ is given in \cite[Theorem 1]{Sullivan1967}.)
	Applying this to $H_{-i}$ also implies 
	\prettyref{eq:q-zero-le-p} and \prettyref{eq:bitmap-zero}.

	\item Of direct relevance is the result of 
Cammarota and
Russo on Bernoulli measures of cosets \cite{CammarotaRusso1991}. 
In particular, 
\eqref{eq:four-point} (and hence \eqref{eq:zero-mode}) follows from
\cite[Proposition~3.1]{CammarotaRusso1991}.

	\item The connection to spin glass stems from the fact that the posterior distribution $\Law(X|HX=s)$ is an Ising model,\footnote{This observation was also used in \cite{Macris2007} and the GKS inequality for Gaussian coupling \cite{Morita2004} is applied to analyze MAP coding of LDPC codes over binary-input Gaussian channel.} with interaction order given by the sparsity of each row of $H$. To see this, 
	let $\sigma_i\triangleq(-1)^{x_i}$ and $\sigma_S \triangleq \prod_{i \in S} \sigma_i$.
Denote the support of the $\ell$th row of $H$ as 
$\partial\ell\triangleq\{i:H_{\ell i}=1\}$. 
For the coset $A_s=\{x:Hx=s\}$, the posterior $P_{A_s}=\Law(X\mid HX=s)$ 
(induced on the $\pm 1$ spins) is the hard-constraint Gibbs measure
\begin{equation}\label{eq:posterior-hard-ising}
  P_{A_s}(\sigma)
  \propto
  \exp\left\{\lambda\sum_{i=1}^n\sigma_i\right\}
  \prod_{\ell=1}^k
  \Indc_{\left\{\sigma_{\partial\ell}
    =(-1)^{s_\ell}\right\}},
\end{equation}
where $\lambda \triangleq\frac12\ln\frac{1-p}{p}\geq0$.
Note that $P_{A_s}$ is the limit of $P_{A_s}^{(\beta)}(\sigma)$ as $\beta\to \infty$:
\[
P_{A_s}^{(\beta)}(\sigma)
 \propto 
\exp\left\{\lambda \sum_{i=1}^n\sigma_i
+ \beta \sum_{\ell=1}^k(-1)^{s_\ell} \sigma_{\partial\ell}\right\}	
\]

Crucially, for $s=0$ and $C = A_0=\ker(H)$, $P_C^{(\beta)}$ is \emph{ferromagnetic} so is the limiting $P_C$. (This is why the posterior for the zero coset is special.)
Thus 	the correlation inequality of 
	Griffiths-Kelly-Sherman (GKS) \cite{Griffiths1967,KellySherman1968} is applicable: 
	$\Cov_{P_C}(\sigma_{A},\sigma_{B}) \geq 0$ for any subset $A,B$. 
	From here, we can arrive at 
	\begin{equation}
	P_C(X=0) \geq \prod_{i=1}^n P_C(X_i=0)
	\label{eq:PC0}
	\end{equation}
	which combined with \prettyref{eq:smallKL} implies \prettyref{eq:zero-atom-divergence-bound} and hence \prettyref{lmm:mass-energy}.
	
	To show \prettyref{eq:PC0}, consider the linear subspace $C_i = C \cap \{X_1=0, \ldots, X_i=0\} = C_{i-1}\cap \{X_i=0\}$.
	For each $j\leq i-1$, applying GKS to $P_{C_{j-1}}$ yields 
	$P_{C_{j-1}}(X_i=0,X_j=0)\geq P_{C_{j-1}}(X_i=0) P_{C_{j-1}}(X_j=0)$, namely, 
	$P_{C_j}(X_i=0)\geq P_{C_{j-1}}(X_i=0)$. Repeated application yields
	$P_{C_{i-1}}(X_i=0)\geq P_{C}(X_i=0)$. Multiplying this from $i=1$ to $n$ gives \prettyref{eq:PC0}.
	(Note that \prettyref{eq:PC0} also follows from the ``conditional
G-regularity'' in \cite[Theorem~1.1]{CammarotaRusso1991}.)

\end{itemize}

\subsection{Equivalent formulation for channel coding}\label{sec:converse-achievability}

In view of the operational duality between source and channel coding \cite{Massey1978}, \prettyref{thm:finite-converse} may be phrased equivalently for the binary symmetric channel (BSC) with crossover probability $p$ as follows.
%The converse also has a channel-coding interpretation. 
Let $\calC$ be a linear codebook of dimension $m$ in $\FF_2^n$, whose communication rate is
\[
  R\triangleq\frac{m}{n} \quad \text{bits/channel use}.
\]
Choose a codeword $c$ uniformly from $\calC$.
The BSC channel output is $Y=c+Z$, where the noise $Z\sim P$ is independent of $c$. Let $P_{\mathrm b}$ denote the average bit error rate under the optimal bitwise MAP decoding of the codeword. 
The capacity of the binary symmetric channel is
\[
  C_{\mathrm{BSC}}\triangleq1-h(p).
\]
If $R>C_{\mathrm{BSC}}$, reliable communication is impossible and hence $P_b$ is bounded away from zero. The following is a precise lower bound for linear codes:

\begin{corollary}\label{cor:channel-form}
For $C_{\mathrm{BSC}}\leq R \leq 1$,
\begin{equation}\label{eq:channel-form}
  P_{\mathrm b}
  \geq
  \frac{p}{h(p)}
  \bigl(R-C_{\mathrm{BSC}}\bigr).
\end{equation}
\end{corollary}

\begin{proof}
Let $H\in\FF_2^{(n-m)\times n}$ be a parity-check matrix of the code $\calC$ with $\calC = \ker(H)$.
For every observation $y\in\FF_2^n$, Bayes' rule and the uniform prior on $c$ give
\begin{equation}\label{eq:channel-noise-posterior}
  \Prob(Z=z\mid Y=y)
  \propto
  P(z)\Indc_{\{y+z\in\calC\}}
  =P(z)\Indc_{\{Hz=Hy\}}.
\end{equation}
Thus the conditional law of $Z$ given $Y=y$ is the Bernoulli product measure conditioned on the syndrome $HZ=Hy$. Since $c=y+Z$, bitwise MAP decoding of $c$ is equivalent to bitwise MAP reconstruction of $Z$ from $HZ$, and therefore $P_{\mathrm b}=D(H)$. Applying \prettyref{thm:finite-converse} with $k=n-m$ gives
\[
  1-R
  \geq h(p)\left(1-\frac{P_{\mathrm b}}{p}\right),
\]
which is \eqref{eq:channel-form}.
\end{proof}

\section*{Declaration of AI use}
The first version of the proof was discovered by GPT-5.6 Sol in an interactive process guided by the author, who subsequently simplified the program and developed the self-contained proof presented here. 
In hindsight, and thanks to a literature search aided by Codex, it turns out that several of the proof ingredients have appeared in, or can be deduced from, the prior literature (see \prettyref{sec:literature-connections} for discussion). 
The author wrote the paper and assumes ful responsibility for its technical content.

\section*{Acknowledgment}

The author is grateful to Sergio Verd\'u for bringing this problem and the paper \cite{Massey1978} to his attention in 2007.

%\bibliographystyle{alpha}
%\bibliography{ancheta_massey_references}

\end{document}